%% file: paper.tex
\documentclass[11pt]{article}

\usepackage{latexsym}
\usepackage{graphicx}  
\usepackage{dcolumn}   
\usepackage{bm}        
\usepackage{amsmath,amssymb,amsthm,mathrsfs,amsfonts,dsfont} 
\usepackage{arydshln}
\usepackage{mathtools} \mathtoolsset{showonlyrefs=true}
\usepackage[margin=2cm]{geometry}
\usepackage{caption}
\usepackage{subcaption}
\usepackage{authblk}
\usepackage[parfill]{parskip}
\usepackage{float}
\usepackage{enumerate}
\usepackage{mathpazo}
\usepackage{cite}
\begingroup
\makeatletter
\@for\theoremstyle:=definition,remark,plain\do{%
	\expandafter\g@addto@macro\csname th@\theoremstyle\endcsname{%
		\addtolength\thm@preskip\parskip
	}%
}
\endgroup

\usepackage{hyperref}
\hypersetup{colorlinks, citecolor=magenta, filecolor=black, linkcolor=blue, urlcolor=black}

\newcommand{\ket}[1]{\left| #1 \right>} 
\newcommand{\bra}[1]{\left< #1 \right|} 
\newcommand{\proj}[1]{\left| #1 \rangle\langle #1 \right|} 

\newcommand{\R}{\mathbb{R}}

\newcommand{\Prob}{\mathbb{P}}

\newcommand{\LieG}[1]{\mathcal{#1}}

\newcommand{\eqn}[1]{eq. \eqref{eqn:#1}}

\newcommand{\iden}{\mathds{1}}

\newcommand{\eqgaph}{\hspace{0.5cm}}

\newcommand{\cond}{\mathsf{COND}}
\newcommand{\qcond}{\mathsf{QCOND}}
\newcommand{\pcond}{\mathsf{PCOND}}
\newcommand{\pqcond}{\mathsf{PQCOND}}

\newcommand{\samp}{\mathsf{SAMP}}
\newcommand{\qsamp}{\mathsf{QSAMP}}
\newcommand{\compare}{\textsc{Compare}}
\newcommand{\qcompare}{\textsc{QCompare}}

\newcommand{\estprobtwo}{\textsc{AddEstProb}}
\newcommand{\estprobmul}{\textsc{MulEstProb}}
\newcommand{\qestprobmul}{\textsc{MulEstProbQCond}}
\newcommand{\qestprobadd}{\textsc{AddEstProbQCond}}
\newcommand{\I}{\mathcal{B}}
\newcommand{\E}{\mathbb{E}}
\newcommand{\modlessgap}[1]{\hspace{-0.6em}\mod{#1}}
\newcommand{\rXY}{r_{X,Y}}
\newcommand{\rYX}{r_{Y,X}}
\newcommand{\tilderXY}{\tilde{r}_{X,Y}}

\newcommand{\comment}[1]{\textcolor{red}{(#1)}}

\renewcommand{\comment}[1]{\ignorespaces}

\theoremstyle{plain}
\newtheorem{thm}{Theorem}[section]

\theoremstyle{plain}
\newtheorem{defn}[thm]{Definition}

\theoremstyle{remark}

\theoremstyle{plain}
\newtheorem{lem}[thm]{Lemma}

\theoremstyle{plain}

\theoremstyle{plain}
\newtheorem{cor}[thm]{Corollary}

\theoremstyle{plain}
\newtheorem{problem}[thm]{Problem}

\theoremstyle{plain}

\floatstyle{ruled}
\newfloat{alg}{t}{alg}
\floatname{alg}{Algorithm}

\DeclareMathOperator{\tr}{Tr}

\begin{document}

\title{Quantum conditional query complexity}

\input{authors.tex}
\date{}

\maketitle

\begin{abstract}
We define and study a new type of quantum oracle, the quantum conditional oracle, which provides oracle access to the conditional probabilities associated with an underlying distribution.
Amongst other properties, we (a) obtain speed-ups over the best known quantum algorithms for identity testing, equivalence testing and uniformity testing of probability distributions; (b) study the power of these oracles for testing properties of boolean functions, and obtain an algorithm for checking whether an $n$-input $m$-output boolean function is balanced or $\epsilon$-far from balanced; and (c) give a sub-linear algorithm, requiring $\tilde{O}(n^{3/4}/\epsilon)$ queries, for testing whether an $n$-dimensional quantum state is maximally mixed or not.
\end{abstract}


\newcommand{\step}[1]{Step~\ref{alg:step#1}}
\newcommand{\steps}[2]{Steps~\ref{alg:step#1} and \ref{alg:step#2}}
\newcommand{\stepsi}[2]{Steps~\ref{alg:step#1}--\ref{alg:step#2}}
\newcommand{\case}[1]{Case~\ref{enum:qcompareproof#1}}
\newcommand{\twocases}[2]{Cases~\ref{enum:qcompareproof#1} and~\ref{enum:qcompareproof#2}}
\newcommand{\algoutput}[1]{\textup{\textsf{#1}}}
\newcommand{\extreme}{\algoutput{Extreme}}
\newcommand{\low}{\algoutput{Low}}
\newcommand{\high}{\algoutput{High}}
\newcommand{\equal}{\algoutput{Equal}}
\newcommand{\far}{\algoutput{Far}}
\newcommand{\specialcell}[2][c]{\begin{tabular}[#1]{@{}l@{}}#2\end{tabular}}

\section{Introduction} \label{s:intro}


One of the fundamental challenges in statistics is to infer information about properties of large datasets as efficiently as possible. 
This is becoming increasingly important as we collect progressively more data about our world and our lives.
Often one would like to determine a certain property of the collected data while having no physical ability to access all of it.
This can be formalised as the task of \emph{property testing}: determining whether an object has a certain property, or is `far' from having that property, ideally minimising the number of inspections of it.
There has been an explosive growth in recent years in this field~\cite{goldreich_property_1998, goldreich_property_2010, blais_property_2012}, and particularly in the sub-field of \emph{distribution testing}, in which one seeks to learn information about a data set by drawing samples from an associated probability distribution.

The \emph{classical conditional sampling oracle} ($\cond$)~\cite{acharya_chasm_2014,canonne_testing_2015,chakraborty_power_2016} grants access to a distribution $D$ such that one can draw samples not only from $D$, but also from $D_S$, the conditional distribution of $D$ restricted to an arbitrary subset $S$ of the domain. Such oracle access reveals a separation between the classical query complexity of identity testing (i.e. whether an unknown distribution $D$ is the same as some known distribution $D^*$), which takes a constant number of queries, and equivalence testing (i.e. whether two unknown distributions $D_1$ and $D_2$ are the same), which requires $\Omega(\sqrt{\log\log N})$ queries, where $N$ is the size of the domain~\cite{acharya_chasm_2014}.
In this paper we introduce a natural quantum version of the $\cond$ oracle (see Definition~\ref{defn:qcond} below) and study its computational power. 

More specifically, we will consider the $\pcond$ (pairwise-$\cond$) oracle, which only accepts query subsets $S$ of cardinality $2$ or $N$, and introduce the $\pqcond$ (pairwise-$\qcond$) oracle. While being rather restricted in comparison to the full $\cond$ and $\qcond$ oracles, they nevertheless offer significant advantages over the standard sampling oracles.

\subsection{Results} \label{ss:results}
\textbf{Quantum algorithms for property testing problems.}
We study the following property testing tasks for classical probability distributions and present efficient algorithms for their solution using our $\pqcond$ oracle. We compare our results with previously known bounds for the standard quantum sampling oracle $\qsamp$ and the classical $\pcond$ oracle.
\begin{enumerate}
\item {\it Uniformity Test:} Given a distribution $D$ and a promise that $D$ is either the uniform distribution $\mathcal{A}$ or $|D-\mathcal{A}|\geq\epsilon$, where $|\cdot|$ is the $L_1$-norm, decide which of the options holds.
\item {\it Known-distribution Test:} Given a fixed distribution $D^*$ and a promise that either $D=D^*$ or $|D-D^*|\geq\epsilon$, decide which of the options holds.
\item {\it Unknown-distribution Test:} Given two distributions $D_1$ and $D_2$ and a promise that either $D_1=D_2$ or $|D_1-D_2|\geq\epsilon$, decide which of the options holds.
\item {\it Distance from uniformity:} Given a distribution $D$ and the uniform distribution $\mathcal{A}$, estimate $\hat{d} = |D-\mathcal{A}|$.
\end{enumerate}
The query complexities for the above problems are listed in Table~\ref{tbl:resultssummary}, with our new results given in the last column. The notation $\tilde{O}(f(N, \epsilon))$ denotes $O(f(N, \epsilon) \log^k f(N, \epsilon))$ for some $k$, i.e. logarithmic factors are hidden.

\begin{table}[h] 
	\centering
	\def\arraystretch{2}
	\begin{tabular}{>{\bfseries}m{2in} >{\centering}m{1.5in} >{\centering}m{1.5in} >{\centering\arraybackslash}m{1.5in}}
		\hline
			\textbf{Task} &
			\textbf{Standard quantum oracle ($\qsamp$)} &
			\textbf{$\pcond$ oracle}~\cite{canonne_testing_2015} &
			\textbf{$\pqcond$ oracle ~~~~~~}\textit{[this work]}\\
		\hline
		Uniformity Test & $O\left(\frac{N^{1/3}}{\epsilon^{4/3}}\right)$~\cite{bravyi_quantum_2011} & $\tilde{O}\left(\frac{1}{\epsilon^2}\right)$ & $\tilde{O}\left(\frac{1}{\epsilon}\right)$\\
		Known-distribution Test & $\tilde{O}\left(\frac{N^{1/3}}{\epsilon^5}\right)$~\cite{chakraborty_new_2010}& $\tilde{O}\left[ \left( \frac{\log N}{\epsilon} \right)^4 \right]$ & $\tilde{O}\left[ \left( \frac{\log N}{\epsilon} \right)^3 \right]$  \\
		Unknown-distribution Test & $O\left(\frac{N^{1/2}}{\epsilon^6}\right)$~\cite{bravyi_quantum_2011} & $\tilde{O} \left[ \left( \frac{\log^2 N}{\epsilon^7} \right)^3 \right]$  & $\tilde{O} \left[ \left( \frac{\log^2 N}{\epsilon^7} \right)^2 \right]$ \\
		Distance from uniformity \vspace{1ex} & $O\left(\frac{N^{1/2}}{\epsilon^6}\right)$~\cite{bravyi_quantum_2011} \vspace{1ex} & $\tilde{O}\left(\frac{1}{\epsilon^{20}}\right)$ \vspace{1ex} & $\tilde{O} \left( \frac{1}{\epsilon^{13}} \right)$ \vspace{1ex} \\
		\hline
	\end{tabular}
	\caption{Query complexity for property testing problems using three different access models: the standard quantum oracle ($\qsamp$), the $\pcond$ oracle, and our $\pqcond$ oracle.}
	\label{tbl:resultssummary}
\end{table}

\textbf{Testing properties of boolean functions.}
A slight modification of the $\pqcond$ oracle will allow for the testing of properties of boolean functions.

Given $f:\{0, 1\}^n\rightarrow \{0, 1\}^m$ with $n \geq m$, define $F_i \coloneqq |\{x\in \{0,1\}^n : f(x) = i\}|/2^n$ for $i\in\{0, 1\}^m$. The function $f$ is promised to be either:
	\begin{itemize}
		\item a balanced function, i.e. $F_i = \frac{1}{2^m} \hspace{1em}\forall i\in\{0, 1\}^m$; or
		\item $\epsilon$-far from balanced, i.e. $\sum_{i\in\{0, 1\}^m} |F_i-\frac{1}{2^m}| \geq \epsilon$.
	\end{itemize}
Provided we have $\pqcond$ access to $f$, we present a quantum algorithm that decides which of these is the case using $\tilde{O}(1/\epsilon)$ queries.

\textbf{Quantum spectrum testing.}
We consider a quantum cloud-based computation scenario in which one or more small, personal quantum computers query a central quantum data hub $\mathcal{Q}$ to deduce properties of a dataset.

Suppose this hub has access to an $n$-dimensional mixed state $\rho$ (in the form of a full classical description, or simply through having access to a large number of copies of $\rho$), and a personal quantum computer $\mathcal{P}$ wishes to determine properties of $\rho$. Each query from $\mathcal{P}$ is effected as follows:
\begin{enumerate}
	\item $\mathcal{P}$ prepares a state of three registers: the first is classical and describes a basis $\I=\{\ket{b_i}\}_{i\in [n]}$; the second and third are quantum, prepared in a state of $\mathcal{P}$'s choosing. $\mathcal{P}$ sends the three registers to $\mathcal{Q}$.
	\item Given these registers, $\mathcal{Q}$ provides $\pqcond$ access to the distribution $D^{(\rho,\I)}_{[n]}(i)=Tr(\rho \proj{b_i}) = \bra{b_i}\rho\ket{b_i}$, with the quantum registers being the input and output registers for the $\pqcond$ query. $\mathcal{Q}$ finally returns the quantum registers to $\mathcal{P}$.
\end{enumerate}

We consider the problem of testing whether or not $\rho$ is the maximally mixed state. More formally, it is promised that one of the following holds:
	\begin{itemize}
		\item $\|\rho-\iden/n\|_1 = 0$, i.e. $\rho$ is the maximally mixed state; or
		\item $\|\rho-\iden/n\|_1 \geq \epsilon$, i.e. $\rho$ is $\epsilon$-far from the maximally mixed state,
	\end{itemize}
	where $\|\cdot\|_1$ is the trace norm\footnote{For an $(n\times n)$ matrix $A$, $\|A\|_1 = \tr\sqrt{AA^\dagger} =\sum_{i\in[n]} a_i$, where the $a_i$ are the singular values of $A$.}.
The task for $\mathcal{P}$ is to decide which is the case.

We present a quantum algorithm to decide the above problem that uses $\tilde{O} (n^{3/4}/\epsilon)$ $\pqcond$ queries.

This problem has also been studied in a setting where $\mathcal{P}$ has access to an unlimited number of copies of the state $\rho$~\cite{odonnell_quantum_2015}, and an optimal algorithm was presented that used $\tilde{O} (n/\epsilon^2)$ copies of the state.

\subsection{Motivation}

The conditional access model is versatile and well-suited to a wide range of practical applications, a few of which are mentioned below.

\textbf{Lottery machine.}
A \emph{gravity pick lottery machine} works as follows:
$N$ balls, numbered $1, \dots, N$, are dropped into a spinning machine, and after a few moments a ball is released.
One might wish to determine whether or not such a machine is fair, i.e. whether or not a ball is released uniformly at random.
A distribution testing algorithm would correctly decide between the following options (assuming that one is guaranteed to be true) with high probability:

\begin{itemize}
	\item The lottery machine is fair and outputs $i$ with probability $1/N$;
	\item The lottery machine is $\epsilon$-far from uniform.
\end{itemize}
In this example, access to a $\cond$ oracle is equivalent to being able to choose which balls are allowed into the spinner.
Classically, it is known that $\Theta(N^{1/2}/\epsilon^4)$ queries~\cite{batu2001testing} to the $\samp$ oracle are required to determine whether or not a distribution generated by such a lottery machine is uniform. However, given access to the corresponding quantum oracle, $\qsamp$, only $O(N^{1/3}/\epsilon^{4/3})$ queries are required~\cite{bravyi_quantum_2011}.
Using the $\pqcond$ oracle we are able to achieve this with $\tilde{O}(1/\epsilon)$ queries.

\textbf{Predicting movie preferences.} Suppose we had a large enough amount of data about two movies, $A$ and $B$, in order to access the joint probability distribution $D$ describing how many people watch these movies on any given day. 
One would like to find out if people watching movie $A$ are more likely to watch movie $B$. More generally, we ask: is $D$ a product of two independent distributions, or are viewings of movie $A$ correlated with viewings of movie $B$? 
The distribution testing algorithm can be used to decide between the following options:

\begin{itemize}
	\item $D$ is independent; i.e. $D$ is a product of two distributions, $D=D^{(A)} \times D^{(B)}$;
	\item $D$ is $\epsilon$-far from independent; i.e. it is $\epsilon$-far from every product distribution.
\end{itemize}

\textbf{Other tests.} There is a wide range of other informative property tests, including:
\begin{itemize}
	\item Checking if two unknown distributions are identical.
	\item Checking if a distribution is identical to a known reference distribution.
	\item Estimating the support size of a distribution.
	\item Estimating the entropy of a distribution.
\end{itemize}



Many of these have been extensively studied in the classical~\cite{batu_testing_2010, valiant_power_2011, canonne_testing_2015, canonne_testing_2014, diakonikolas_testing_2015, guha_streaming_2006, chan_optimal_2014, chakraborty_new_2010} and quantum~\cite{bravyi_quantum_2011, montanaro_survey_2013} literature, and near-optimal bounds have often been placed on the number of queries required to solve the respective problems.



\subsection{Outline}
In Section~\ref{s:prelim} we introduce notation and define our quantum conditional oracles. In Section~\ref{s:efficientcomparison} we prove a lemma that is subsequently used to obtain our main technical tool---the $\qcompare$ function, which efficiently compares conditional probabilities of a distribution.
In Section~\ref{s:probdist} we apply it to obtain new, efficient query complexity bounds for property testing of probability distributions. In Section~\ref{s:boolean} we test properties of boolean functions, before presenting a quantum spectrum test in Section~\ref{s:qspectrumtesting}.

\section{Preliminaries and Notation} \label{s:prelim}
Let $D$ be a probability distribution over a finite set $[N] \coloneqq \{0, 1, \dots, N-1\}$, where $D(i) \geq 0$ is the weight of the element $i \in [N]$.
Furthermore, if $S\subseteq[N]$, then $D(S)=\sum_{i\in S} D(i)$ is the weight of the set $S$.
If $D(S)>0$, define $D_S$ to be the conditional distribution, i.e. $D_S(i) \coloneqq D(i)/D(S)$ if $i \in S$ and $D_S(i) = 0$ if $i \notin S$.

Below, we recall the definitions of the classical and quantum sampling oracles, and subsequently define the classical and quantum conditional sampling oracles.

\begin{defn}[Classical Sampling Oracle~\cite{canonne_testing_2015}]
	Given a probability distribution $D$ over $[N]$, we define the \emph{classical sampling oracle} $\samp_D$ as follows: each time $\samp_D$ is queried, it returns a single $i \in [N]$, where the probability that element $i$ is returned is $D(i)$.
\end{defn}

\begin{defn}[Quantum Sampling Oracle~\cite{bravyi_quantum_2011}] \label{defn:qsamp}
	Given a probability distribution $D$ over $[N]$, let $T \in \mathbb{N}$ be some specified integer, and assume that $D$ can be represented by a mapping $O_D:[T]\rightarrow[N]$ such that for any $i\in[N]$, $D(i)$ is proportional to the number of elements in the pre-image of $i$, i.e. $D(i) = |\{ t \in [T] : O_D(t) = i \}| / T$.
	In other words, $O_D$ labels the elements of $[T]$ by $i \in [N]$, and the $D(i)$ are the frequencies of these labels, and are thus all rational with denominator $T$.
	
	Then each query to the \emph{quantum sampling oracle} $\qsamp_D$ applies the unitary operation $U_D$, described by its action on basis states:
	\[ U_D \ket{t} \ket{\beta} = \ket{t} \ket{\beta + O_D(t)\modlessgap{N}}. \]
	In particular,
	\[ U_D \ket{t} \ket{0} = \ket{t} \ket{O_D(t)}. \]
	
	As an example, note that querying with a uniformly random $t \in [T]$ in the first register will result in $i \in [N]$ in the second register with probability $D(i)$.
\end{defn}

\begin{defn}[Classical Conditional Sampling Oracle~\cite{canonne_testing_2015}]
	Given a probability distribution $D$ over $[N]$ and a set $S\subseteq [N]$ such that $D(S)>0$, we define the \emph{classical conditional sampling oracle} $\cond_D$ as follows: each time $\cond_D$ is queried with query set $S$, it returns a single $i \in [N]$, where the probability that element $i$ is returned is $D_S(i)$.
\end{defn}

We are now ready to define a new \emph{quantum conditional sampling oracle}, a quantum version of $\cond_D$.

\begin{defn}[Quantum Conditional Sampling Oracle] \label{defn:qcond}
	Given a probability distribution $D$ over $[N]$, let $T \in \mathbb{N}$ be some specified integer, and assume that there exists a mapping $O_D: \mathcal{P}([N]) \times [T] \rightarrow [N]$, where $\mathcal{P}([N])$ is the power set of $[N]$, such that for any $S \subseteq [N]$ with $D(S) > 0$ and any $i \in [N]$, $D_S(i) = |\{ t \in [T] : O_D(S, t) = i \}| / T$.
	
	Then each query to the \emph{quantum conditional sampling oracle} $\qcond_D$ applies the unitary operation $U_D$, defined below.
	
	$U_D$ acts on 3 registers: 
	\begin{itemize}
		\item The first consists of $N$ qubits, whose computational basis states label the $2^N$ possible query sets $S$;
		\item The second consists of $\log T$ qubits that describe an element of $[T]$; and
		\item The third consists of $\log N$ qubits to store the output, an element of $[N]$.
	\end{itemize}
	The action of the oracle on basis states is
	\[ U_D \ket{S} \ket{t} \ket{\beta} = \ket{S} \ket{t} \ket{\beta + O_D(A, t) \modlessgap{N}}. \]
	
	In particular,
	\[ U_D \ket{S} \ket{t} \ket{0} = \ket{S} \ket{t} \ket{O_D(A, t)}. \]
\end{defn}

\textit{Remark:} Note that querying $\qcond_D$ with query set $S=[N]$ is equivalent to a query to $\qsamp_D$.

The $\pcond_D$ oracle, described in~\cite{canonne_testing_2015}, only accepts query subsets $S$ of cardinality $2$ or $N$. Below we define its quantum analogue, the $\pqcond_D$ oracle.

\begin{defn}[Pairwise Conditional Sampling Oracle]
	The $\pqcond_D$ oracle is equivalent to the $\qcond_D$ oracle, with the added requirement that the query set $S$ must satisfy $|S|=2$ or $N$, i.e. the distribution can only be conditioned over pairs of elements or the whole set.
\end{defn}

\section{Efficient comparison of conditional probabilities}\label{s:efficientcomparison}

In this section we first prove a lemma to improve the dependency on success probability for a general probabilistic algorithm. We subsequently use this result to prove our main technical tool, the $\qcompare$ algorithm, which compares conditional probabilities of a distribution, and is crucial to our improved property tests.

\subsection{Improving dependence on success probability}

The following lemma, proved in Section~\ref{ss:probabilitydependence}, provides a general method for improving the dependence between the number of queries made and the success probability of the algorithm. 

\begin{lem}\label{lem:probabilitydependence}
	Suppose an algorithm $\textsc{Alg1}(\xi, \epsilon, \delta)$ $( \epsilon > 0, \delta \in (0, 1] )$ outputs an (additive) approximation to $f(\xi)\in \R$. More formally, suppose it outputs $\tilde{f}(\xi)$ such that $\Prob[|\tilde{f}(\xi) - f(\xi)| \leq \epsilon] \geq 1 - \delta$ using $M(\xi, \epsilon, \delta)$ queries to a classical/quantum oracle, for some function $M$.

	Then there exists an algorithm $\textsc{Alg2}(\xi, \epsilon, \delta)$ that makes $\Theta ( M( \xi, \epsilon, \frac{1}{10}) \log(1/\delta) )$ queries to the same oracle and outputs $\tilde{f}(\xi)$ such that $ \Prob[|\tilde{f}(\xi) - f(\xi)| \leq \epsilon] \geq 1 - \delta$, i.e. the dependence of the number of queries on the success probability can be taken to be $\log(1/\delta)$.
\end{lem}



Applying this lemma to Theorem 5 of~\cite{bravyi_quantum_2011} gives an exponential improvement, from $1/\delta$ to $\log(1/\delta)$, in the dependence on the success probability given there. This is summarised in the theorem below.

\begin{thm}
	\label{thm:estprobtwo}
	There exists a quantum algorithm $\estprobtwo(D,S,M)$ that takes as input a distribution $D$ over $[N]$, a set $S\subset [N]$ and an integer $M$. The algorithm makes exactly $M$ queries to the $\qsamp_D$ oracle and outputs $\tilde{D}(S)$, an approximation to $D(S)$, such that $\Prob[|\tilde{D}(S) - D(S)| \leq \epsilon] \geq 1-\delta$ for all $\epsilon > 0$ and $\delta\in (0, 1]$ satisfying
	\[ M \geq c\log(1/\delta) \max\left( \frac{\sqrt{D(S)}}{\epsilon}, \frac{1}{\sqrt{\epsilon}} \right),  \]	
	where $c=O(1)$ is some constant.	
\end{thm}

A multiplicative version Theorem~\ref{thm:estprobtwo} follows straightforwardly:
\begin{thm}
	\label{thm:estprobmul}
	There exists a quantum algorithm $\estprobmul(D,S,M)$ that takes as input a distribution $D$ over $[N]$, a set $S\subset [N]$ and an integer $M$. The algorithm makes exactly $M$ queries to the $\qsamp_D$ oracle and outputs $\tilde{D}(S)$, an approximation to $D(S)$, such that $\Prob[\tilde{D}(S) \in [1-\epsilon, 1+\epsilon] D(S) ] \geq 1-\delta$ for all $\epsilon,\delta \in (0, 1]$ satisfying
	\[ M \geq \frac{c \log(1/\delta)}{\epsilon \sqrt{D(S)}}, \]
	where $c=O(1)$ is some constant.
\end{thm}

Access to the $\qcond_D$ oracle effectively gives us access to the oracle $\qsamp_{D_S}$ for any $S\subseteq [N]$, and this allows us to produce stronger versions of Theorems~\ref{thm:estprobtwo} and~\ref{thm:estprobmul}:

\begin{thm}
	\label{thm:estprobaddqcond}
	There exists a quantum algorithm $\qestprobadd(D,S,R,M)$ that takes as input a distribution $D$ over $[N]$, a set $S\subseteq [N]$ with $D(S) > 0$, a subset $R \subset S$ and an integer $M$. The algorithm makes exactly $M$ queries to the $\qcond_D$ oracle and outputs $\tilde{D}_S(R)$, an approximation to $D_S(R)$, such that $\Prob[|\tilde{D}_S(R) - D_S(R)| \leq \epsilon] \geq 1-\delta$ for all $\epsilon > 0$ and $\delta\in (0, 1]$ satisfying
	\[ M \geq c\log(1/\delta) \max\left( \frac{\sqrt{D_S(R)}}{\epsilon}, \frac{1}{\sqrt{\epsilon}} \right),  \]	
	where $c=O(1)$ is some constant.
\end{thm}

\begin{thm}
	\label{thm:estprobmulqcond}
	There exists a quantum algorithm $\qestprobmul(D,S,R,M)$ that takes as input a distribution $D$ over $[N]$, a set $S\subseteq [N]$ with $D(S) > 0$, a subset $R \subset S$ and an integer $M$. The algorithm makes exactly $M$ queries to the $\qcond_D$ oracle and outputs $\tilde{D}_S(R)$, an approximation to $D_S(R)$, such that $\Prob[\tilde{D}_S(R) \in [1-\epsilon,1+\epsilon] D_S(R)] \geq 1-\delta$ for all $\epsilon,\delta \in (0, 1]$ satisfying
	\[ M \geq \frac{c \log(1/\delta)}{\epsilon \sqrt{D_S(R)}}, \]
		where $c=O(1)$ is some constant.
\end{thm}

\subsection{The \textmd{\textsc{QCompare}} algorithm}
An important routine used in many classical distribution testing protocols (see~\cite{canonne_testing_2015}) is the $\compare$ function, which outputs an estimate of the ratio $\rXY \coloneqq D(Y)/D(X)$ of the weights of two disjoint subsets $X,Y\subset [N]$ over $D$. As stated in Section 3.1 of~\cite{canonne_testing_2015}, if $X$ and $Y$ are disjoint, $D(X \cup Y)>0$, and $1/K \leq \rXY \leq K$ for some integer $K\geq 1$, the algorithm outputs $\tilderXY \in [1 - \eta, 1 + \eta] \rXY$ with probability at least $1-\delta$ using only $\Theta(K \log(1/\delta) / \eta^2)$ $\cond_D$ queries. Surprisingly, the number of queries is independent of $N$, the size of the domain of the distribution.

Here we introduce a procedure called $\qcompare$ that uses the $\qcond_D$ oracle and subsequent quantum operations to perform a similar function to $\compare$, achieving the same success probability  and bound on the error with $\Theta(\sqrt{K} \log(1/\delta) / \eta)$ queries.

We now use \qestprobadd~and \qestprobmul~to create the \qcompare~procedure.

\begin{alg}[H]
	\caption{$\qcompare(D, X, Y, \eta, K, \delta)$}\label{alg:qcompare}
	
	\textbf{Input:} $\qcond$ access to a probability distribution $D$ over $[N]$, disjoint subsets $X,Y\subset [N]$ such that $D(X \cup Y) > 0$, parameters $K\geq 1$, $\eta \in (0, \frac{3}{8K})$, and $\delta\in(0, 1]$.
	
	\begin{enumerate}
		\item Set $M = \Theta\left(\frac{\sqrt{K} \log(1/\delta)}{\eta}\right)$. \label{alg:step1}
		\item Set $\tilde{w}_+(X) = \qestprobadd(D, X \cup Y, X, M)$. \label{alg:step2a}
		\item Set $\tilde{w}_+(Y) = \qestprobadd(D, X \cup Y, Y, M)$. \label{alg:step2b}
		\item Set $\tilde{w}_\times(X) = \qestprobmul(D, X \cup Y, X, M)$. \label{alg:step3a}
		\item Set $\tilde{w}_\times(Y) = \qestprobmul(D, X \cup Y, Y, M)$. \label{alg:step3b}
		\item Check that $\tilde{w}_+(X) \leq \frac{3K}{3K+1} - \frac{\eta}{3}$. If the check fails, return \low~and exit.\label{alg:step4a}
		\item Check that $\tilde{w}_+(Y) \leq \frac{3K}{3K+1} - \frac{\eta}{3}$. If the check fails, return \high~and exit.\label{alg:step4b}
		\item Return $\tilderXY = \frac{\tilde{w}_\times(Y)}{\tilde{w}_\times(X)}$. \label{alg:step5}
	\end{enumerate}
\end{alg}

\begin{thm}\label{thm:qcompare}
	Given the input as described, \qcompare~(Algorithm~\ref{alg:qcompare}) outputs \low, \high, or a value $\tilderXY > 0$, and satisfies the following:
	\begin{enumerate}
		\item If $1/K \leq \rXY \leq K$, then with probability at least $1-\delta$ the procedure outputs a value $\tilderXY \in [1-\eta,1+\eta]\rXY$;
		\item If $\rXY > K$ then with probability at least $1-\delta$ the procedure outputs either \high~or a value $\tilderXY \in [1-\eta,1+\eta]\rXY$;
		\item If $\rXY < 1/K$ then with probability at least $1-\delta$ the procedure outputs either \low~or a value $\tilderXY \in [1-\eta,1+\eta]\rXY$.
	\end{enumerate}
	The procedure performs $\Theta\left(\frac{\sqrt{K} \log(1/\delta)}{\eta}\right)$ $\qcond_D$ queries on the set $X \cup Y$ via use of \qestprobadd~and \qestprobmul.
\end{thm}
The proof of this theorem is given in Section~\ref{ss:qcompareproof}.
\comment{Check all $\Theta$'s and other Big-$O$ notation!}

\section{Property testing of probability distributions} \label{s:probdist}
We now apply our results to obtain new algorithms for a number of property testing problems. 

\begin{cor} \label{cor:unif}
	Let $\mathcal{A}^{(N)}$ be the uniform distribution on $[N]$ (i.e. $\mathcal{A}^{(N)}(i) = 1/N, i\in[N]$). Given $\pqcond$ access to a probability distribution $D$ over $[N]$, there exists an algorithm that uses $\tilde{O}(1/\epsilon)$ $\pqcond_D$ queries and decides with probability at least $2/3$ whether
	\begin{itemize}
		\item $|D-\mathcal{A}^{(N)}| = 0$ (i.e. $D=\mathcal{A}^{(N)}$), or
		\item $|D-\mathcal{A}^{(N)}| \geq \epsilon$,
	\end{itemize}
	provided that it is guaranteed that one of these is true. Here $|\cdot|$ is the $L_1$-norm\footnote{For two distributions $D_1$ and $D_2$ over $[N]$, $|D^{(1)} - D^{(2)}| = \sum_{i\in[N]} |D^{(1)}(i) - D^{(2)}(i)|$.}.
\end{cor}
\begin{proof}
	We replace the calls to $\compare$ with the corresponding calls to $\qcompare$ in Algorithm 4 of~\cite{canonne_testing_2015}. For this method, calls to $\qcompare$ only require conditioning over pairs of elements, and hence the $\pqcond_D$ oracle may be used instead of $\qcond_D$.
\end{proof}

\textit{Remark:} The corresponding classical algorithm (Algorithm 4 in~\cite{canonne_testing_2015}) uses $\tilde{O}(1/\epsilon^2)$ $\pcond_D$ queries. The authors also show  (Section 4.2 of~\cite{canonne_testing_2015}) that any classical algorithm making $\cond_D$ queries must use $\Omega(1/\epsilon^2)$ queries to solve this problem with bounded probability. Thus the above quantum algorithm is quadratically more efficient than any classical $\cond$ algorithm.

\begin{cor}
	Given the full specification of a probability distribution $D^*$ (i.e. a \emph{known} distribution) and $\pqcond$ access to a probability distribution $D$, both over $[N]$, there exists an algorithm that uses $\tilde{O}\Big(\frac{\log^3 N}{\epsilon^3}\Big)$ $\pqcond_{D}$ queries and decides with probability at least $2/3$ whether
	\begin{itemize}
		\item $|D-D^*| = 0$ (i.e. $D=D^*$), or
		\item $|D-D^*| \geq \epsilon$,
	\end{itemize}
	provided that it is guaranteed that one of these is true.
\end{cor}
\begin{proof}
	We replace the calls to $\compare$ with the corresponding calls to $\qcompare$ in Algorithm 5 of~\cite{canonne_testing_2015}.
\end{proof}

\textit{Remark:} The corresponding classical algorithm (Algorithm 5 in~\cite{canonne_testing_2015}) uses $\tilde{O}\Big(\frac{\log^4 N}{\epsilon^4}\Big)$ $\pcond_D$ queries.

\begin{cor} \label{cor:unknown}
	Given $\pqcond$ access to probability distributions $D^{(1)}$ and $D^{(2)}$ over $[N]$, there exists an algorithm that decides, with probability at least $2/3$, whether
	\begin{itemize}
		\item $|D^{(1)}-D^{(2)}| = 0$ (i.e. $D^{(1)}=D^{(2)}$), or
		\item $|D^{(1)}-D^{(2)}| \geq \epsilon$,
	\end{itemize}
	provided that it is guaranteed that one of these is true. The algorithm uses $\tilde{O}\Big(\frac{\log^4 N}{\epsilon^{14}}\Big)$ $\pqcond_{D^{(1)}}$ and $\pqcond_{D^{(2)}}$ queries.
\end{cor}
\begin{proof}
	We replace the calls to $\compare$ with the corresponding calls to $\qcompare$ in Algorithm 9 of~\cite{canonne_testing_2015}. As a by-product of this process, the function \textsc{Estimate-Neighborhood} (Algorithm 2 in~\cite{canonne_testing_2015}), using $\tilde{O}\Big(\frac{\log(1/\delta)}{\kappa^2 \eta^4 \beta^3 \delta^2}\Big)$ $\pcond$ queries, is replaced by an algorithm \textsc{QEstimate-Neighborhood}, which uses $\tilde{O}\Big(\frac{\log(1/\delta)}{\kappa \eta^3 \beta^2 \delta}\Big)$ $\pqcond$ queries.
\end{proof}

\textit{Remark:} This is to be compared with Algorithm 9 in~\cite{canonne_testing_2015}, which uses $\tilde{O}\Big(\frac{\log^6 N}{\epsilon^{21}}\Big)$ $\pcond_{D^{(1)}}$ and $\pcond_{D^{(2)}}$ queries.

\begin{cor}
	Given $\pqcond$ access to a probability distribution $D$ over $[N]$, there exists an algorithm that uses $\tilde{O}(1/\epsilon^{13})$ queries and outputs a value $\hat{d}$ such that $|\hat{d}-|D-\mathcal{A}^{(N)}|| = O(\epsilon)$.
\end{cor}
\begin{proof}
	We replace the calls to $\compare$ with the corresponding calls to $\qcompare$ in Algorithm 11 of~\cite{canonne_testing_2015}. In addition, we trivially replace all queries to the $\samp_D$ oracle with queries to $\pqcond_D$ with query set $[N]$. As a by-product of this process, the function \textsc{Find-Reference} (Algorithm 12 in~\cite{canonne_testing_2015}), using $\tilde{O}(1/\kappa^{20})$ $\pcond$ and $\samp$ queries, is replaced by an algorithm \textsc{QFind-Reference}, which uses $\tilde{O}(1/\kappa^{13})$ $\pqcond$ queries.
\end{proof}

\textit{Remark:} The corresponding classical algorithm (Algorithm 11 in~\cite{canonne_testing_2015}) uses $\tilde{O}(1/\epsilon^{20})$ queries.

\section{Property testing of Boolean functions}\label{s:boolean}

The results in Section~\ref{s:probdist} can be applied to test properties of Boolean functions. One of the more important challenges in field of cryptography is to determine whether or not a given boolean function is `balanced'. We give an algorithm to solve this problem with a constant number of $\pqcond$ queries.

Consider a function $f: \{0,1\}^n \rightarrow \{0,1\}^m$, for $n,m\in \mathbb{N}$ with $n\geq m$. If $m=1$, we might consider the following problem:

\begin{problem}[Constant-balanced problem]
	Given $f:\{0, 1\}^n\rightarrow \{0, 1\}$, decide whether
	\begin{itemize}
		\item $f$ is a balanced function, i.e. $|\{x\in \{0,1\}^n : f(x) = 0\}|/2^n = |\{x\in \{0,1\}^n : f(x) = 1\}|/2^n = \frac{1}{2}$, or
		\item $f$ is a constant function, i.e. $f(x)=0 \hspace{0.5em} \forall x\in \{0,1\}^n$ or $f(x)=1 \hspace{0.5em} \forall x\in \{0,1\}^n$,
	\end{itemize}
	provided that it is guaranteed that $f$ satisfies one of these conditions.	
\end{problem}

With standard quantum oracle access to $f$, this problem can be solved exactly with one query, through use of the Deutsch-Jozsa algorithm~\cite{cleve1998quantum,deutsch_rapid_1992}. Consider the following extension of this problem:

\begin{problem} \label{problem:dj2}
	Given $f:\{0, 1\}^n\rightarrow \{0, 1\}$, write $F_i \coloneqq |\{x\in \{0,1\}^n : f(x) = i\}|/2^n$. Decide whether
	\begin{itemize}
		\item $f$ is a balanced function, i.e. $F_0 = F_1 = \frac{1}{2}$, or
		\item $f$ is $\epsilon$-far from balanced, i.e. $|F_0-\frac{1}{2}| + |F_1-\frac{1}{2}| = 2|F_0-\frac{1}{2}| \geq \epsilon$,
	\end{itemize}
	provided that it is guaranteed that $f$ satisfies one of these conditions.	
\end{problem}

This problem can be solved with bounded probability by querying $f$ several times.
In addition, it can be solved using the $\qsamp$ oracle.
To understand how this works, set $T=2^n, N=2, O_D=f$ in Definition~\ref{defn:qsamp} so that $D(i) = F_i$. Then Theorem~\ref{thm:estprobtwo} can be used to estimate $D(0) = F_0$ to error $\epsilon/3$ with probability $1-\delta$ using $O(\log(1/\delta)/\epsilon)$ queries.

Now we consider an even more general problem:

\begin{problem}
	Given $f:\{0, 1\}^n\rightarrow \{0, 1\}^m$, write $F_i \coloneqq |\{x\in \{0,1\}^n : f(x) = i\}|/2^n$. Decide whether
	\begin{itemize}
		\item $f$ is a balanced function, i.e. $F_i = \frac{1}{2^m} \hspace{1em}\forall i\in\{0, 1\}^m$, or
		\item $f$ is $\epsilon$-far from balanced, i.e. $\sum_{i\in\{0, 1\}^m} |F_i-\frac{1}{2^m}| \geq \epsilon$,
	\end{itemize}
	provided that it is guaranteed that $f$ satisfies one of these conditions.	
\end{problem}

By allowing $\pqcond$ access to $f$, this can be solved in $\tilde{O}(1/\epsilon)$ queries. In what sense do we allow $\pqcond$ access to $f$? We relate $f$ to a probability distribution by setting $N=2^m$, $D_{[N]}(i) = F_i$, and using the definition of $D_S(i)$ given at the start of Section~\ref{s:prelim}. The problem is then solved by an application of the algorithm presented in Corollary~\ref{cor:unif}.


\section{Quantum Spectrum Testing} \label{s:qspectrumtesting}

Recall the quantum cloud-based computation scenario presented in Section~\ref{ss:results}. 
It is easy to see that for any basis $\I$, $D_{[n]}^{\iden/n,\I} = \mathcal{A}^{(n)}$, where $\mathcal{A}^{(n)}$ is the uniform distribution over $[n]$. Then for any state $\rho$,
\begin{itemize}
	\item if $\|\rho-\iden/n\|_1 = 0$, then $\left| D_{[n]}^{\rho, \I} - \mathcal{A}^{(n)} \right| = 0$ for any basis $\I$;
	\item if $\|\rho-\iden/n\|_1 \geq \epsilon$, perhaps we can choose a basis $\I$ such that $\left| D_{[n]}^{\rho, \I} - \mathcal{A}^{(n)} \right| \geq \nu(\epsilon, n)$, for some function $\nu$.
\end{itemize}
Corollary~\ref{cor:unif}, with distance parameter $\nu(\epsilon,n)$, could then be used to distinguish between these two options.

As the first case above is immediate, we henceforth assume that $\|\rho-\iden/n\|_1 \geq \epsilon$. In order to simplify the analysis, we assume that $n$ is even, let $\Delta = \rho - \iden / n$, and introduce
\begin{equation}
	\delta^{(\I)} \coloneqq \left| D_{[n]}^{\rho, \I} - \mathcal{A}^{(n)} \right| = \sum_{i\in[n]} |\bra{b_i} \Delta \ket{b_i}|.
\end{equation}

Let $\tilde{\I} = \left\{ \ket{\tilde{b}_i} \right\}_{i\in[n]}$ be the eigenbasis of $\Delta$, and let $d_i \coloneqq \bra{\tilde{b}_i} \Delta \ket{\tilde{b}_i}$, $i \in [n]$ be the eigenvalues. Thus, $\Delta = \sum_{i \in [n]} d_i \proj{\tilde{b}_i}$. Note that $\tr \Delta = \sum_{i\in[n]} d_i = 0$, and also $\eta \coloneqq \|\rho-\iden/n\|_1 = \|\Delta\|_1 = \sum_{i\in[n]} |d_i| \geq \epsilon$.

Now suppose we choose a basis $\I = \{ \ket{b_i} \}_{i \in [n]}$ uniformly at random, i.e. we choose $W\in \LieG{U}(n)$ uniformly at random according to the Haar measure, and set $\ket{b_i} = W\ket{\tilde{b}_i}$. Then
\begin{equation}
	\delta^{(\I)} = \sum_{i \in [n]} |\bra{b_i} \Delta \ket{b_i}| = \sum_{i \in [n]} \left| \sum_{j \in [n]} |W_{ji}|^2 d_j \right|
\end{equation}

The triangle inequality then gives
\begin{align}
	\delta^{(\I)} \geq \Bigg| \sum_{j\in[n]} \Bigg( \sum_{i\in[n]} |W_{ji}|^2 \Bigg) d_j \Bigg| = \Bigg| \sum_{j\in[n]} d_j \Bigg| = 0; \eqgaph \delta^{(\I)} \leq \sum_{j\in[n]} \Bigg( \sum_{i\in[n]} |W_{ji}|^2 \Bigg) |d_j| = \eta. \label{eqn:Edeltabounds}
\end{align}
	
Let $v^{(i)}_j = |W_{ji}|^2$, introduce the vector $V^{(i)} = (v^{(i)}_0, \dots, v^{(i)}_{n-1})$, and write $d=(d_0, \dots, d_{n-1})$. Then
\begin{equation}
	\delta^{(\I)} = \sum_{i\in[n]} |V^{(i)} \cdot d|.
\end{equation}

We now make use of Sykora's theorem~\cite{sykora_quantum}, which states that if $W$ is chosen uniformly at random according to the Haar measure on $\LieG{U}(n)$, then the vector $V^{(i)}$, for any $i$, is uniformly distributed over the probability simplex
\begin{equation}
	T_n = \{ (v_0, \dots, v_{n-1}) : v_i \in [0, 1], \, \textstyle \sum_{i\in [n]} v_i = 1 \}.
\end{equation}

Since all of the $V^{(i)}$'s have the same distribution, we see that
\begin{equation}
	\E \left( \delta^{(\I)} \right) = n \E (|V\cdot d|),
\end{equation}
where $V$ is a generic $V^{(i)}$.

The following lemma allows us to relate a lower bound on $\E \left( \delta^{(\I)} \right)$ to a lower bound on $\Prob[\delta^{(\I)} \geq \lambda]$, for some $\lambda$.

\begin{lem} \label{lem:exptoprob}
	\begin{equation}
		\Prob \left[ \delta^{(\I)} \geq \lambda \right] \geq \frac{1}{\eta} \left( \E \left( \delta^{(\I)} \right) - \lambda^2 \right)
	\end{equation}
\end{lem}
\begin{proof}
	Let $p=p(\mu)$ be the probability density function for $\delta^{(\I)}$. As noted in \eqn{Edeltabounds}, $0 \leq \delta^{(\I)} \leq \eta$. Thus, for $\lambda \in [0, \eta]$ we can write
	\begin{align}
		\E \left( \delta^{(\I)} \right) &= \int_0^\eta \mu p(\mu) \hspace{0.2em} d\mu \\
		&= \int_0^\lambda \mu p(\mu) \hspace{0.2em} d\mu + \int_\lambda^\eta \mu p(\mu) \hspace{0.2em} d\mu \\
		&\leq \int_0^\lambda \lambda \cdot 1 \hspace{0.2em} d\mu + \int_\lambda^\eta \eta p(\mu) \hspace{0.2em} d\mu \\
		&= \lambda^2 + \eta \Prob \left[ \delta^{(\I)} \geq \lambda \right].
	\end{align}
	
	Rearranging the inequality gives the result.
\end{proof}

\textit{Remark:} One might consider using Chebyshev's inequality~\cite{grimmett2014probability} to place a bound on $\Prob[\delta^{(\I)} \geq \lambda]$. The above lemma achieves a tighter bound, however, which is necessary for the remainder of this section.

We now write $\E(|V\cdot d|)$ as an integral over the probability simplex $T_n$. We have
\begin{equation}
	\E(f(V)) = \int_{T_n} f(V) dV \coloneqq (n-1)! \int_{v_0=0}^1 \cdots \int_{v_{n-1}=0}^1 \delta(1-\textstyle \sum_{i\in[n]} v_i) f(V) \hspace{0.2em} dv_0 \cdots dv_{n-1} \label{eqn:simplexintegral}
\end{equation}
where $dV= (n-1)! \hspace{0.2em} \delta(1-\sum_{i\in [n]} v_i) \hspace{0.2em} dv_0 \cdots dv_{n-1}$ is the normalised measure on $T_n$, defined so that $\E(1) = 1$.

Now note that the integral expression for $\E(|V\cdot d|) = \E(|v_0 d_0 + \cdots v_{n-1} d_{n-1}|)$ is completely symmetric in the $v_i$'s (and hence in the $d_i$'s). Thus, if $\sigma$ is a permutation on $[n]$, we have that
\begin{equation}
	\E(|v_0 d_0 + \cdots v_{n-1} d_{n-1}|) = \E(|v_0 d_{\sigma(0)} + \cdots v_{n-1} d_{\sigma(n-1)}|).
\end{equation}

Using this observation, we can write
\begin{align}
	&\hspace{1.5em} \E(|v_0 d_0 + \cdots v_{n-1} d_{n-1}|) \\
	&= \frac{1}{n} \left[ \E(|v_0 d_{\sigma(0)} + \cdots v_{n-1} d_{\sigma(n-1)}|) + \E(|v_0 d_{\sigma(1)} + \cdots v_{n-1} d_{\sigma(0)}|) + \right. \\
	& \hspace{2cm} \left. + \E(|v_0 d_{\sigma(2)} + \cdots v_{n-1} d_{\sigma(1)}|) + \cdots + \E(|v_0 d_{\sigma(n-1)} + \cdots v_{n-1} d_{\sigma(n-2)}|) \right] \\
	&= \frac{1}{n} \left[ \E(|v_0 d_{\sigma(0)} + \cdots v_{n-1} d_{\sigma(n-1)}|) + \E(|-v_0 d_{\sigma(1)} - \cdots v_{n-1} d_{\sigma(0)}|) + \right. \\
	& \hspace{2cm} \left. + \E(|v_0 d_{\sigma(2)} + \cdots v_{n-1} d_{\sigma(1)}|) + \cdots + \E(|-v_0 d_{\sigma(n-1)} - \cdots v_{n-1} d_{\sigma(n-2)}|) \right] \label{eqn:symmplusminus} \\
	&\geq \frac{1}{n} \E\left[|v_0 (d_{\sigma(0)} - d_{\sigma(1)} + \cdots - d_{\sigma(n-1)}) + v_1 (d_{\sigma(1)} - d_{\sigma(2)} + \cdots - d_{\sigma(0)}) \right. \\
	& \left. \hspace{2cm} + v_2 (d_{\sigma(2)} - d_{\sigma(3)} + \cdots - d_{\sigma(1)}) + \cdots v_{n-1} (d_{\sigma(n-1)} - d_{\sigma(0)} + \cdots - d_{\sigma(n-2)})|\right] \label{eqn:symmtriangle} \\
	&= \frac{1}{n}|d_{\sigma(0)}-d_{\sigma(1)}+d_{\sigma(2)}-\cdots -d_{\sigma(n-1)}| \hspace{0.5em} \E(|v_0-v_1+v_2-\cdots -v_{n-1}|),
\end{align}
where in \eqn{symmplusminus} minus signs are added inside every other expectation (note that $n$ is even), and \eqn{symmtriangle} is derived using the triangle inequality.

Since $\sigma$ was an arbitrary permutation, we can instead write
\begin{equation}
	\E(|V\cdot d|) \geq \frac{1}{n} \left[ \max_{\sigma \in Sym([n])} |d_{\sigma(0)}-d_{\sigma(1)}+d_{\sigma(2)}-\cdots -d_{\sigma(n-1)}| \right] \hspace{0.5em} \E(|v_0-v_1+v_2-\cdots -v_{n-1}|),
\end{equation}
where $Sym([n])$ symmetric group on $[n]$, and hence
\begin{equation}
	\E \left( \delta^{(\I)} \right) \geq M^{(d)} E_n,
\end{equation}
where
\begin{align}
M^{(d)} &\coloneqq \max_{\sigma \in Sym([n])} |d_{\sigma(0)}-d_{\sigma(1)}+d_{\sigma(2)}-\cdots -d_{\sigma(n-1)}|, \label{eqn:Md} \\
E_n &\coloneqq \E(|v_0-v_1+v_2-\cdots -v_{n-1}|). \label{eqn:En}
\end{align}

Evaluation of $M^{(d)}$ and $E_n$ is carried out in Sections~\ref{ss:Md} and~\ref{ss:En}, where we find that $M^{(d)} \geq \frac{1}{2}\eta$ and $E_n \geq \frac{1}{2\sqrt{n}}$. Hence
\begin{equation}
	\E \left( \delta^{(\I)} \right) \geq \frac{\eta}{4\sqrt{n}}. \label{eqn:explowerbound}
\end{equation}

Use of Lemma~\ref{lem:exptoprob} immediately tells us that
\begin{equation}
	\Prob \left[ \delta^{(\I)} \geq \lambda \right] \geq \frac{1}{4\sqrt{n}} - \frac{\lambda^2}{\eta}.
\end{equation}

Setting $\lambda = \frac{\min(1,\epsilon)}{\sqrt{8} n^{1/4}}$ and recalling that $\epsilon \leq \eta$ gives
\begin{equation}
	\Prob \left[ \delta^{(\I)} \geq \frac{\min(1,\epsilon)}{\sqrt{8} n^{1/4}} \right] \geq \frac{1}{4\sqrt{n}} - \frac{\min(1,\epsilon)^2}{8 \eta \sqrt{n}} \geq \frac{1}{4\sqrt{n}} - \frac{1}{8\sqrt{n}} = \frac{1}{8\sqrt{n}}.
\end{equation}

Suppose we repeat this test $k$ times, choosing different bases $\I_1, \dots, \I_k$ uniformly at random according to the Haar measure on $\LieG{U}(n)$. We call $\I$ `good' if $\delta^{(\I)} \geq \frac{\min(1,\epsilon)}{\sqrt{8} n^{1/4}}$. Let $K(k)$ represent the event that at least one of $\I_1, \dots, \I_k$ is `good'. Then
\begin{equation}
	\Prob[K(k)] \geq 1-\left( 1-\frac{1}{8\sqrt{n}} \right)^k.
\end{equation}

Setting $k=32\sqrt{n}$ gives
\begin{equation}
	\Prob[K(32\sqrt{n})] \geq 1-\frac{1}{e^4} \geq \frac{49}{50}.
\end{equation}

\subsection{Executing the algorithm}

The algorithm given in Corollary~\ref{cor:unif} succeeds with probability at least $\frac{2}{3}$.
Suppose we run the algorithm $l$ times in total. Then by using a Chernoff bound (eq. (1) in~\cite{canonne_testing_2015}), it follows that
\begin{itemize}
	\item if the distributions are `equal', $\Prob \left[\mbox{algorithm outputs \equal} \geq \frac{1}{2} l \mbox{ times} \right] \geq 1 - e^{-l/18}$;
	\item if the distributions are `far', $\Prob \left[\mbox{algorithm outputs \far} \geq \frac{1}{2} l \mbox{ times} \right] \geq 1 - e^{-l/18}$.
\end{itemize}

The full algorithm has been set out below.

\begin{alg}[H]
	\caption{$\textsc{MaximallyMixedStateTest}(\rho)$}\label{alg:qspectrum}
	
	\textbf{Input:} $\pqcond$ access to a probability distribution $D_{[n]}^{(\rho, \I)}$ over $[n]$ for any $\I$, as described in Section~\ref{s:qspectrumtesting}, and parameter $\epsilon$. Set $l=128 \log n$.
	
	\begin{enumerate}
		\item Choose $k=32\sqrt{n}$ bases $\I_1, \dots \I_k$ uniformly at random. \label{alg:stepq1}
		\item For each $j=1, \dots, k$, run the algorithm given in Corollary~\ref{cor:unif} on the distribution  $D_{[n]}^{(\rho, \I_j)}$ $l$ times, returning $u_j=1$ if at least $\frac{1}{2} l$ of the runs return \far, and $u_j=0$ otherwise. \label{alg:stepq2}
%
		\item If any $u_j$ is equal to $1$, output \far, otherwise output \equal. \label{alg:stepq3}
	\end{enumerate}
\end{alg}

The analysis of this algorithm is separated into two cases:
\begin{itemize}
	\item $\|\rho-\iden/n\|_1 = 0$:	
	The probability that a particular $u_j$ is equal to $1$ in \step{q2} is less than $e^{-l/18}$. Thus, the probability of the algorithm failing is, by the union bound\footnote{For a countable set of events $A_1, A_2, \dots$, we have that $\Prob\left[\bigcup_i A_i \right] \leq \sum_i \Prob[A_i]$.}, at most $32\sqrt{n} \hspace{0.5em} e^{-l/18} \leq \frac{1}{3}$, and hence the algorithm outputs \equal~with probability at least $\frac{2}{3}$.
	\item $\|\rho-\iden/n\|_1 \geq \epsilon$:
	Suppose that $\I_j$ is `good'. Then with probability at least $1-e^{l/18} \geq \frac{99}{100}$, we get $u_j=1$, and the algorithm will output \far~in \step{q3}. The probability that one of $\I_1, \dots, \I_k$ is `good' is at least $\frac{49}{50}$, and hence the probability that the entire algorithm outputs \far~is at least $0.97 \geq \frac{2}{3}$.
\end{itemize}

Each run of the algorithm given in Corollary~\ref{cor:unif} requires $\tilde{O} (n^{1/4}/\epsilon)$ $\pqcond$ queries if $\epsilon \leq 1$, and hence in total Algorithm~\ref{alg:qspectrum} requires
\begin{equation}
	\tilde{O} \left(k l \frac{n^{1/4}}{\epsilon}\right) = \tilde{O} \left(\frac{n^{3/4}}{\epsilon}\right)
\end{equation}
$\pqcond$ queries.

\section{Discussion}\label{s:discussion}
Quantum conditional oracles give us new insights into the kinds of information that are useful for testing properties of distributions.
In many circumstances such oracles serve as natural models for accessing information.
In addition, they are able to demonstrate separations in query complexity between a number of problems, thereby providing interesting new perspectives on information without trivialising the set-up. We now mention some open questions.

Group testing and pattern matching are further important areas to which our notion of a quantum conditional oracle could be applied. The structure of questions commonly considered there suggest that use of $\pqcond$ would decrease the query complexity dramatically for many practically relevant problems compared to the best known quantum and classical algorithms~\cite{problem_33, de_bonis_optimal_2005, ambainis_efficient_2015, bonis_constraining_2015}.

In our algorithms, we have made particular use of the $\pqcond$ oracle, the quantum analogue of the $\pcond$ oracle. It is noted in~\cite{canonne_testing_2015} that the unrestricted $\cond$ oracle offers significant advantages over the $\pcond$ oracle for many problems, and it is possible that similar improvements could be achieved for some quantum algorithms through use of the unrestricted $\qcond$ oracle.

The algorithm that we present for quantum spectrum testing (Algorithm~\ref{alg:qspectrum}) chooses several bases $\I_1, \dots, \I_k$ independently and uniformly at random. It remains open, however, whether or not a more adaptive approach to choosing bases will yield an algorithm requiring fewer queries.

Our definition of the spectrum testing problem in Section~\ref{s:qspectrumtesting} made use of the trace norm, $\| \cdot \|_1$. One might wonder how the query complexity would be affected if the problem were defined with a different norm, such as the operator norm\footnote{For an $(n\times n)$ matrix $A$, $\|A\|_\infty = \max_{i\in[n]} a_i$, where the $a_i$ are the singular values of $A$.}, $\| \cdot \|_\infty$. Numerical simulations and limited analysis suggest that the probability of picking a `good' basis $\I$ tends to $1$ as $n \rightarrow \infty$, and hence that the number of queries required to distinguish between the two options would be independent of $n$. We leave the proof of this conjecture as an open question.


\appendix
\section*{Appendix}
\section{Efficient comparison of conditional probabilities}
\subsection{Proof of Lemma~\ref{lem:probabilitydependence}}\label{ss:probabilitydependence}

We first state the procedure for $\textsc{Alg2}(\xi, \epsilon, \delta)$ $( \epsilon > 0, \delta \in (0, 1] )$.
\begin{enumerate}
	\item Run $\textsc{Alg1}(\xi, \epsilon, \frac{1}{10})$ $m$ times, where $m=\Theta(\log(1/\delta))$ (and such that $m$ is even). Denote the outputs as $\tilde{f}_1, \dots \tilde{f}_m$, labelled such that $\tilde{f}_1 \leq \cdots \leq \tilde{f}_m$. \label{alg:alg2_1}
	\item Output $\tilde{f}_{m/2}$. \label{alg:alg2_2}
\end{enumerate}

Consider $\textsc{Alg1}(\xi, \epsilon, \frac{1}{10})$, and let $E_1$ be the event that $|\tilde{f}(\xi)-f(\xi)|\leq \epsilon$, which is equivalent to the event that $\tilde{f}(\xi) \in [f(\xi)-\epsilon, f(\xi)+\epsilon]$. Then we have that $\Prob[E_1]\geq \frac{9}{10}$.

Let $Y$ be a random variable that takes the value $1$ if $E_1$ occurs during a run of $\textsc{Alg1}(\xi, \epsilon, \frac{1}{10})$, and $0$ otherwise. Let $Y_1, \dots, Y_m \sim Y$ be i.i.d. random variables. Let $E_2$ be the event that at least $\frac{8}{10}m$ of the $Y_i$ output $1$ (i.e. the event that $E_1$ occurs at least $\frac{8}{10}m$ times).

Using a Chernoff bound (here we use eq. (1) in~\cite{canonne_testing_2015}), it is easy to see that $\Prob[E_2] \geq 1 - \exp(-\tfrac{1}{50} m)$.

Setting $m=\Theta(\log(1/\delta))$ and rounding $m$ up to the nearest multiple of $2$ then gives $\Prob[E_2] \geq 1 - \delta$.

Thus, we see that, with probability at least $1-\delta$, Step~\ref{alg:alg2_1} results in $\tilde{f}_1 \leq \dots \leq \tilde{f}_m$ such that $|\tilde{f}_i - f(\xi)| \leq \epsilon$ for at least $\frac{8}{10}m$ values of $i\in \{1, \dots, m\}$.
Henceforth we assume that $E_2$ occurs. Now consider $\tilde{f}_{m/2}$. Suppose $\tilde{f}_{m/2}\notin [f(\xi)-\epsilon, f(\xi)+\epsilon]$. Then one of the two following statements must hold:
\begin{itemize}
	\item $\tilde{f}_{m/2} < f(\xi) - \epsilon$. Since $\tilde{f}_1\leq \dots \leq \tilde{f}_{m/2}$, we have that $\tilde{f}_1, \dots, \tilde{f}_{m/2} \notin [f(\xi)-\epsilon, f(\xi)+\epsilon]$, which contradicts the statement of $E_2$;
	\item $\tilde{f}_{m/2} > f(\xi) + \epsilon$. Since $\tilde{f}_{m/2} \leq \dots \leq \tilde{f}_m$, we have that $\tilde{f}_{m/2}, \dots, \tilde{f}_m \notin [f(\xi)-\epsilon, f(\xi)+\epsilon]$, which contradicts the statement of $E_2$.
\end{itemize}
Hence we conclude that $\tilde{f}_{m/2} \in [f(\xi)-\epsilon, f(\xi)+\epsilon]$.
\qed

\textit{Remark:} It is worth noting that the method used in the above proof could also apply to different kinds of algorithms, and not just the specific algorithm $\textsc{Alg1}$.

\subsection{Proof of Theorem~\ref{thm:qcompare}} \label{ss:qcompareproof}

We prove this case-by-case. We define the shorthand $w(X) \coloneqq D_{X \cup Y}(X) = D(X) / D(X \cup Y)$, $w(Y) \coloneqq D_{X \cup Y}(Y) = D(Y) / D(X \cup Y)$ and note that $\rXY = w(Y) / w(X)$. In addition, since $w(X) + w(Y) = 1$, it is straightforward to show the following inequalities for a constant $T \geq 1$:
\begin{equation}
	\begin{aligned}
		\rXY \geq \frac{1}{T} \hspace{0.5em}&\implies\hspace{0.5em} w(X) \leq \frac{T}{T+1}, \hspace{0.5em} w(Y) \geq \frac{1}{T+1} \\
		\rXY \leq \frac{1}{T} \hspace{0.5em}&\implies\hspace{0.5em} w(X) \geq \frac{T}{T+1}, \hspace{0.5em} w(Y) \leq \frac{1}{T+1} \\
		\rXY \geq T \hspace{0.5em}&\implies\hspace{0.5em} w(X) \leq \frac{1}{T+1}, \hspace{0.5em} w(Y) \geq \frac{T}{T+1} \\
		\rXY \leq T \hspace{0.5em}&\implies\hspace{0.5em} w(X) \geq \frac{1}{T+1}, \hspace{0.5em} w(Y) \leq \frac{T}{T+1}		
		\label{eqn:rhoinequalities}
	\end{aligned}
\end{equation}
The strict versions of these inequalities also hold true.

\begin{enumerate}
	\item $\bm{1/K \leq \rXY \leq K}$ \label{enum:qcompareproof1}
	
	\textit{In this case we wish our algorithm to output $\tilderXY \in [1-\eta, 1+\eta] \rXY$.} \vspace{2ex}
	
	From \eqn{rhoinequalities}, we immediately have that
	\begin{equation} \frac{1}{K+1} \leq w(X), w(Y) \leq \frac{K}{K+1}. \label{eqn:wxwybounds1} \end{equation}
	
	\steps{2a}{2b} use \qestprobadd~to estimate $w(X)$ and $w(Y)$ to within additive error $\eta/3$ with probability at least $1-\delta/4$. As stated in Theorem~\ref{thm:estprobaddqcond}, this requires
	\[ \Theta\left( \max\left( \frac{\sqrt{w(X)}}{\eta}, \frac{1}{\sqrt{\eta}} \right) \log(1/\delta) \right) = \Theta\left(\frac{\log(1/\delta)}{\eta}\right) \]
	queries to $\qcond_D$, where the equality is due to the fact that $w(X)\leq 1$, and thus $M$ (defined in Algorithm~\ref{alg:qcompare}) queries suffice.
	
	\step{3a} uses \qestprobmul~to estimate $w(X)$ to within multiplicative error $\eta/3$ with probability at least $1-\delta/4$. From Theorem~\ref{thm:estprobmulqcond}, we clearly require
	\[ \Theta\left( \frac{\log(1/\delta)}{\eta \sqrt{w(Y)}} \right) = \Theta\left(\frac{\sqrt{K} \log(1/\delta)}{\eta}\right) \]
	queries to $\qcond_D$ in order to achieve these, where the equality is due to \eqn{wxwybounds1}, and thus $M$ queries suffice. \step{3b} requires the same number of queries.
	
	With a combined probability of at least $1-\delta$, \stepsi{2a}{3b} all pass, and produce the following values:
	\begin{align}
		\tilde{w}_+(X) &\in [w(X)-\eta/3, w(X)+\eta/3], \\
		\tilde{w}_+(Y) &\in [w(Y)-\eta/3, w(Y)+\eta/3], \\
		\tilde{w}_\times(X) &\in [1-\eta/3, 1+\eta/3] w(X), \\
		\tilde{w}_\times(Y) &\in [1-\eta/3, 1+\eta/3] w(Y).
	\end{align}
	
	From \eqn{wxwybounds1}, we see that
	\[ \tilde{w}_+(X), \tilde{w}_+(Y) \leq \frac{K}{K+1} + \frac{\eta}{3} < \frac{3K}{3K+1} - \frac{\eta}{3}, \]
	where the final inequality is due to the algorithm's requirement that $\frac{\eta}{3} < \frac{1}{8K}$.
	
	Thus, the checks in \steps{4a}{4b} pass, and \step{5} gives us
	\[ \tilderXY \in [1-\eta, 1+\eta] \rXY. \]
	\item This is split into 4 sub-cases.
	\begin{enumerate}[a)]
		\item $\bm{K < \rXY}$ \label{enum:qcompareproof2a}
		\begin{enumerate}[i)]
			\item $\bm{3K < \rXY}$ \label{enum:qcompareproof2ai}
			
			\textit{In this case we wish our algorithm to output \high.} \vspace{2ex}
			
			From \eqn{rhoinequalities} we have that
			\begin{equation} w(X) < \frac{1}{3K+1}, \hspace{0.5em} w(Y) > \frac{3K}{3K+1}. \label{eqn:wxwybounds2ai} \end{equation}
			
			As in \case{1}, \steps{2a}{2b} allow us to gain
			\begin{align}
				\tilde{w}_+(X) &\in [w(X)-\eta/3, w(X)+\eta/3], \\
				\tilde{w}_+(Y) &\in [w(Y)-\eta/3, w(Y)+\eta/3],
			\end{align}
			with combined probability at least $1-\delta/2$. (We henceforth assume that we have gained such values.)
			
			Using \eqn{wxwybounds2ai} it is easy to show that $\tilde{w}_+(X) < \frac{3K}{3K+1}-\frac{\eta}{3}$ and that $\tilde{w}_+(Y) > \frac{3K}{3K+1}-\frac{\eta}{3}$. Hence the check in \step{4a} passes, but the check in \step{4b} fails, and the algorithm outputs \high~and exits.
			
			\item $\bm{K < \rXY \leq 3K}$ \label{enum:qcompareproof2aii}
			
			\textit{In this case we wish our algorithm to either output \high~or output $\tilderXY \in [1-\eta, 1+\eta] \rXY$.} \vspace{2ex}
			
			From \eqn{rhoinequalities}, we have that
			\begin{equation}
				\frac{1}{3K+1} \leq w(X) < \frac{1}{K+1}, \hspace{0.5em} \left(\frac{1}{3K+1} <\right) \frac{K}{1+K} < w(Y) \leq \frac{3K}{3K+1}. \label{eqn:wxwybounds2aii}
			\end{equation}
			
			Thus, with $\Theta(\sqrt{K}\log(1/\delta)/\eta)$ queries, as in \case{1}, we gain
			\begin{align}
				\tilde{w}_+(X) &\in [w(X)-\eta/3, w(X)+\eta/3], \\
				\tilde{w}_+(Y) &\in [w(Y)-\eta/3, w(Y)+\eta/3], \\
				\tilde{w}_\times(X) &\in [1-\eta/3, 1+\eta/3] w(X), \\
				\tilde{w}_\times(Y) &\in [1-\eta/3, 1+\eta/3] w(Y),
			\end{align}
			with combined probability at least $1-\delta$. (We henceforth assume that we have gained such values.)
			
			Using \eqn{wxwybounds2aii}, we see that $\tilde{w}_+(X) < \frac{3K}{3K+1}-\frac{\eta}{3}$, and thus \step{4a} will pass.
			
			Assuming the check in \step{4b} passes, \step{5} will output $\tilderXY \in [1-\eta, 1+\eta] \rXY$.
			
			However, given the upper bound for $w(Y)$ in \eqn{wxwybounds2aii}, it is possible to have $\tilde{w}_+(Y) > \frac{3K}{3K+1}-\frac{\eta}{3}$, causing the check in \step{4b} to fail and the algorithm to output \high.
		\end{enumerate}

		\item $\bm{\rXY < 1/K}$  \label{enum:qcompareproof2b}
		\begin{enumerate}[i)]
			\item $\bm{\rXY < 1/(3K)}$  \label{enum:qcompareproof2bi}
			
			This is equivalent to the condition that $3K < \rYX$, and thus follows the same argument as \case{2ai}, with $X$ and $Y$ interchanged and an output of \low~instead of \high.
			\item $\bm{1/(3K) \leq \rXY < 1/K}$ \label{enum:qcompareproof2bii}
			
			This is equivalent to the condition that $K < \rYX \leq 3K$, and thus follows the same argument as \case{2aii}, with $X$ and $Y$ interchanged and an output of \low~instead of \high.
		\end{enumerate}
	\end{enumerate}
\end{enumerate}
\qed

\section{Quantum Spectrum Testing}

\subsection{Evaluating $M^{(d)}$} \label{ss:Md}

\textit{This section provides a lower bound for the quantity $M^{(d)}$, as defined in \eqn{Md}.}

Let $D^+$ be the set of non-negative $d_i$'s, labelled such that $d^+_0 \geq d^+_1 \geq \cdots$, and similarly let $D^-$ be the set of negative $d_i$'s, labelled such that $d^-_0 \leq d^-_1 \leq \cdots$. w.l.o.g. suppose $|D^-| \geq |D^+|$.

Let $|D^+|=\frac{n}{2}-k$, where $k\leq \frac{n}{2}$. Thus $|D^-| = \frac{n}{2} - k$. Note that $\sum_i d_i^+ = -\sum_i d_i^- = \frac{1}{2} \eta$.

We now define $\sigma$ so that the following statements are true:
\begin{itemize}
	\item $d_{\sigma(1)} = d_0^-, d_{\sigma(3)} = d_1^-, \dots, d_{\sigma(n-1)} = d^-_{\frac{n}{2}-1}$;
	\item $d_{\sigma(0)} = d_0^+, d_{\sigma(2)} = d_1^+, \dots, d_{\sigma(n-2k-2)} = d^+_{\frac{n}{2}-k-1}$;
	\item $d_{\sigma(n-2k)}, d_{\sigma(n-2k+2)}, \dots, d_{\sigma(n-2)}$ can be filled with the remaining members of $D^-$.
\end{itemize}

Then
\begin{itemize}
	\item $d_{\sigma(0)} + d_{\sigma(2)} + \cdots + d_{\sigma(n-2k-2)} = \frac{1}{2} \eta$;
	\item $d^-_0, \dots, d^-_{\frac{n}{2}-1} \leq d^-_{\frac{n}{2}-1} \implies -d_{\sigma(1)} - d_{\sigma(3)} - \cdots - d_{\sigma(n-1)} \geq -\frac{n}{2} d^-_{\frac{n}{2}-1}$;
	\item $d^-_{\frac{n}{2}}, \dots, d^-_{\frac{n}{2}+k-1} \geq d^-_{\frac{n}{2}-1} \implies d_{\sigma(n-2k)} + d_{\sigma(n-2k+2)} + \cdots + d_{\sigma(n-2)} \geq k d^-_{\frac{n}{2}-1}$.
\end{itemize}

Hence
\begin{equation}
	|d_{\sigma(0)} - d_{\sigma(1)} + d_{\sigma(2)} - \cdots - d_{\sigma(n)}| \geq \left| \frac{1}{2} \eta + \left( k-\frac{n}{2} \right) d^-_{\frac{n}{2}} \right| \geq \frac{1}{2} \eta,
\end{equation}
where the final inequality follows since $k \leq \frac{n}{2}$ and $d^-_{\frac{n}{2}} < 0$.

Thus $M^{(d)} \geq \frac{1}{2} \eta$.

\subsection{Evaluating $E_n$} \label{ss:En}

\textit{This section provides a lower bound for the quantity $E_n$, as defined in \eqn{En}.}

To evaluate $E_n$ we will use the Hermite-Genocchi Theorem (Theorem 3.3 in~\cite{atkinson2008introduction}), which relates integrals over the probability simplex to associated divided differences.

The divided difference of $n$ points $(x_0, f(x_0)), \dots, (x_{n-1}, f(x_{n-1}))$ is defined by
\begin{equation}
	f[x_0, \dots, x_{n-1}] \coloneqq \sum_{j\in[n]} \frac{f(x_j)}{\prod_{k \neq j} (x_j-x_k)}, \label{eqn:divideddifferences}
\end{equation}
where limits are taken if any of the $x_j$ are equal. It can be shown that for repeated points (see Exercise 4.6.6 in~\cite{schatzman_numerical_2002})
\begin{equation}
	f[\underbrace{x_0, \dots, x_0}_{(r_0+1)\mbox{ times}}, \underbrace{x_1, \dots, x_1}_{(r_1+1)\mbox{ times}}, x_2, \dots, x_{n-1}] = \frac{1}{r_0! r_1!} \, \frac{\partial^{r_0+r_1}}{\partial x_0^{r_0} \partial x_1^{r_1}} \, f[x_0, x_1, x_2, \dots, x_{n-1}], \label{eqn:schatzman}
\end{equation}
where $x_0, \dots, x_{n-1} \in \R$ are distinct.

Now, the Hermite-Genocchi Theorem states that
\begin{equation}
	f[x_0, \dots, x_{n-1}] = \frac{1}{(n-1)!} \int_{T_n} f^{(n-1)}(v_0 x_0 + \cdots v_{n-1} x_{n-1}) \hspace{0.5em} dV, \label{eqn:h-gthm}
\end{equation}
where we recall that $dV=(n-1)! \hspace{0.44em} \delta(1-\sum_{i\in [n]} v_i) \hspace{0.4em} dv_0 \cdots dv_{n-1}$.

In order to evaluate $E_n$, we set $f^{(n-1)}(\xi) = (n-1)! |\xi|$. Thus
\begin{equation}
	f(\xi) = \left\{ \begin{array}{lr}
		\frac{1}{n} \xi^n & \xi \geq 0 \\
		-\frac{1}{n} \xi^n & \xi < 0
	\end{array} \right.
\end{equation}
and $E_n = f[1, -1, 1, -1, \dots, 1, -1]$.

Let $m=\frac{n}{2}-1$ (i.e. $n=2m+2$). Then by \eqn{schatzman} we have that
\begin{equation}
	E_{2m+2} = \frac{1}{m!^2} \left. \partial_0^m \partial_1^m f[x_0, x_1] \right|_{x_0=-1, x_1=1},
\end{equation}
where we have used the notation $\partial_i \equiv \frac{\partial}{\partial x_i}$.

In the neighbourhood of $x_0=-1, x_1=1$, we have (by \eqn{divideddifferences})
\begin{equation}
	f[x_0, x_1] = -\frac{1}{2m+2} \frac{x_0^{2m+2} + x_1^{2m+2}}{x_0 - x_1},
\end{equation}
and thus
\begin{equation}
E_{2m+2} = -\frac{1}{2m+2} \frac{1}{m!^2} \left. A \right|_{x_0=-1, x_1=1}, \label{eqn:AinE_n}
\end{equation}
where
\begin{equation}
	A = \partial_0^m \partial_1^m \left( \frac{x_0^{2m+2} + x_1^{2m+2}}{x_0 - x_1} \right).
\end{equation}

We see that
\begin{align}
	A &= \partial_1^m \partial_0^m \left( \frac{x_0^{2m+2}}{x_0 - x_1} \right) - \partial_0^m \partial_1^m \left( \frac{x_1^{2m+2}}{x_1 - x_0} \right) \\
	&= \partial_1^m \partial_0^m \left( \frac{x_0^{2m+2}}{x_0 - x_1} \right) - (\mbox{\textit{same term with $x_0$ and $x_1$ interchanged}}). \label{eqn:A}
\end{align}

We use the Leibniz product rule\footnote{$(uv)^{(m)} = \sum_{k=0}^m \binom{m}{k} u^{(k)} v^{(m-k)}$} to deduce that
\begin{equation}
	\partial_0^m \left(x_0^n \left( \frac{1}{x_0-x_1} \right) \right)
	= \sum_{k=0}^m \binom{m}{k} \left[ \frac{(2m+2)!}{(2m+2-k)!} x_0^{2m+2-k} \right] \left[ \frac{(-1)^{m-k}}{(x_0-x_1)^{m+1-k}} (m-k)! \right],
\end{equation}
and hence that the first term in \eqn{A} is
\begin{align}
	& \hspace{1.5em} \partial_1^m \partial_0^m \left(x_0^n \left( \frac{1}{x_0-x_1} \right) \right) \\
	&= \sum_{k=0}^m \binom{m}{k} \left[ \frac{(2m+2)!}{(2m+2-k)!} x_0^{2m+2-k} \right] \left[ \frac{(-1)^{m-k}}{(x_0-x_1)^{2m+1-k}} (2m-k)! \right] \\
	&= (2m+2)! (-1)^m \sum_{k=0}^m \binom{m}{k} \frac{(-1)^k (2m-k)!}{(2m+2-k)!} \frac{x_0^{2m+2-k}}{(x_0-x_1)^{2m+1-k}} \\
	&=  (2m+2)! (-1)^m (x_0-x_1) \sum_{k=0}^m \binom{m}{k} \frac{(-1)^k}{(2m+2-k)(2m+1-k)} \left( \frac{x_0}{x_0-x_1} \right)^{2m+2-k}.
\end{align}

Substituting this into \eqn{A} and setting $x_0=-1, x_1=1$ gives
\begin{equation}
	A|_{x_0=-1, x_1=1} = -4 (2m+2)! (-1)^m \sum_{k=0}^m \binom{m}{k} \frac{(-1)^k}{(2m+2-k)(2m+1-k)} \left( \frac{1}{2} \right)^{2m+2-k}. \label{eqn:Arestricted}
\end{equation}

Now set
\begin{equation}
	B = (-1)^m \sum_{k=0}^m \binom{m}{k} \frac{(-1)^k}{(2m+2-k)(2m+1-k)} \gamma^{2m+2-k} \label{eqn:B}
\end{equation}
so that
\begin{equation}
	A|_{x_0=-1, x_1=1} = -4(2m+2)! B|_{\gamma=\frac{1}{2}}. \label{eqn:BinA}
\end{equation}

Next, note that
\begin{equation}
	\frac{\partial^2 B}{\partial \gamma^2} = (-1)^m \sum_{k=0}^m \binom{m}{k} (-1)^k \gamma^{2m-k} = \gamma^m \sum_{k=0}^m \binom{m}{k} (-\gamma)^{m-k} = \gamma^m (1-\gamma)^m,
\end{equation}
and thus
\begin{align}
	B|_{\gamma=\frac{1}{2}} &= \int_{z=0}^\frac{1}{2} \int_{\alpha=0}^z \alpha^m (1-\alpha)^m \hspace{0.5em} d\alpha \hspace{0.2em} dz + C \\
	&= \int_{z=0}^\frac{1}{2} B_z(m+1, m+1) \hspace{0.5em} dz + C,
\end{align}
where $B_z(p, q)=\int_0^z \alpha^{p-1} (1-\alpha)^{q-1} \hspace{0.5em} d\alpha$ is the incomplete Beta function. By setting $m=0$ it is easy to deduce that $C=0$.

Now, the indefinite integral of the incomplete Beta function is
\begin{equation}
	\int B_z(p, q) \hspace{0.5em} dz = z B_z(p, q) - B_z(p+1, q),
\end{equation}
and hence we deduce that
\begin{align}
	B|_{\gamma =\frac{1}{2}} &= \frac{1}{2} B_{\frac{1}{2}}(m+1, m+1) - B_{\frac{1}{2}}(m+2, m+1) \\
	&= \int_0^{\frac{1}{2}} \alpha^m (1-\alpha)^m d\alpha - \int_0^{\frac{1}{2}} \alpha^{m+1} (1-\alpha)^m d\alpha \\
	&= \frac{1}{2} \Bigg[ \int_0^{\frac{1}{2}} \alpha^m (1-\alpha)^m \underbrace{(1-2\alpha)}_{=(1-\alpha)-\alpha} \hspace{0.5em} d\alpha \Bigg] \\
	&= \frac{1}{2} \int_0^{\frac{1}{2}} (\alpha^m (1-\alpha)^{m+1} - \alpha^{m+1} (1-\alpha)^m) \hspace{0.5em} d\alpha \\
	&= \frac{1}{2(m+1)} \int_0^{\frac{1}{2}} \frac{d(\alpha^{m+1} (1-\alpha)^{m+1})}{d\alpha} \hspace{0.5em} d\alpha \\
	&= \frac{1}{2(m+1)} [\alpha^{m+1} (1-\alpha)^{m+1}]_0^{1/2} \\
	&= \frac{1}{2^{2m+3} (m+1)}.
\end{align}

Substituting this into \eqn{BinA} and subsequently into \eqn{AinE_n}, we get
\begin{align}
	E_{2m+2} &= -\frac{1}{2m+2} \frac{1}{m!^2} \cdot -4(2m+2)! \cdot \frac{1}{2^{2m+3} (m+1)} \\
	&= \frac{(2m+1)!}{2^{2m+1} m!^2 (m+1)} \\
	&= \frac{2m+1}{m+1} \cdot \frac{(2m)!}{m!^2} \cdot \frac{1}{2^{2m+1}}. \label{eqn:E_n}
\end{align}

Stirling's formula~\cite{robbins1955remark} tells us that for $m\geq 1$
\begin{equation}
	\sqrt{2 \pi} m^{m+\frac{1}{2}} e^{-m} < m! < \sqrt{2\pi} m^{m+\frac{1}{2}} e^{-m} e^{\frac{1}{12}}, \label{eqn:stirling}
\end{equation}
and thus
\begin{equation}
	\frac{(2m)!}{m!^2} > \frac{\sqrt{2 \pi} (2m)^{2m+\frac{1}{2}} e^{-2m}}{2\pi m^{2m+1} e^{-2m} e^{\frac{1}{6}}} = \frac{2^{2m} e^{-\frac{1}{6}}}{\sqrt{m\pi}}.
\end{equation}

Since $\frac{2m+1}{m+1} \geq \frac{3}{2}$, \eqn{E_n} tells us that
\begin{equation}
	E_{2m+2} > \frac{3e^{-\frac{1}{6}}}{4\sqrt{\pi}} \cdot \frac{1}{\sqrt{m}}.
\end{equation}

Replacing $m$ with $\frac{n}{2}-1$, we deduce that
\begin{align}
	E_n > \frac{3e^{-\frac{1}{6}}}{4\sqrt{\pi}} \cdot \frac{1}{\sqrt{\frac{n}{2}-1}} > \frac{3e^{-\frac{1}{6}}}{4\sqrt{\pi}} \cdot \frac{1}{\sqrt{\frac{n}{2}}} = \frac{3e^{-\frac{1}{6}}}{2\sqrt{2\pi}} \cdot \frac{1}{\sqrt{n}} > \frac{1}{2\sqrt{n}}. \label{eqn:E_nbound}
\end{align}

The case when $m=0$ is easily dealt with through direct calculation using \eqn{E_n}, giving $E_2=\frac{1}{2}$. Hence we conclude that \eqn{E_nbound} holds for all positive, even $n$.

\textit{Remark:} Using the asymptotic form of Stirling's formula, it can be shown that $E_n \sim \sqrt{\frac{2}{\pi}} \frac{1}{\sqrt{n}}$ for large $n$.

\section*{Acknowledgements}
\input acknowledgement.tex

\bibliographystyle{plain}
\bibliography{references}

\end{document}

%% file: authors.tex
\author{Imdad S. B. Sardharwalla} 
\author{Sergii Strelchuk}
\author{Richard Jozsa}
\affil{Department of Applied Mathematics and Theoretical Physics, University of Cambridge, Cambridge, CB3 0WA, U.K.}

%% file: acknowledgement.tex
I.S.B.S. thanks EPSRC for financial support. S.S. acknowledges the support of Sidney Sussex College.

%% file: paper.bbl
\begin{thebibliography}{10}

\bibitem{problem_33}
Problem 33: {Group} {Testing} - {Open} {Problems} in {Sublinear} {Algorithms}.
\newblock {\em Sublinear.info}.

\bibitem{acharya_chasm_2014}
Jayadev Acharya, Clement Canonne, and Gautam Kamath.
\newblock A {Chasm} {Between} {Identity} and {Equivalence} {Testing} with
  {Conditional} {Queries}.
\newblock Technical Report 156, 2014.

\bibitem{ambainis_efficient_2015}
Andris Ambainis, Aleksandrs Belovs, Oded Regev, and Ronald de~Wolf.
\newblock Efficient {Quantum} {Algorithms} for ({Gapped}) {Group} {Testing} and
  {Junta} {Testing}.
\newblock {\em arXiv:1507.03126 [quant-ph]}, July 2015.
\newblock arXiv: 1507.03126.

\bibitem{atkinson2008introduction}
Kendall~E Atkinson.
\newblock {\em An introduction to numerical analysis}.
\newblock John Wiley \& Sons, 2008.

\bibitem{batu2001testing}
Tugkan Batu, Eldar Fischer, Lance Fortnow, Ravi Kumar, Ronitt Rubinfeld, and
  Patrick White.
\newblock Testing random variables for independence and identity.
\newblock In {\em Foundations of Computer Science, 2001. Proceedings. 42nd IEEE
  Symposium on}, pages 442--451. IEEE, 2001.

\bibitem{batu_testing_2010}
Tugkan Batu, Lance Fortnow, Ronitt Rubinfeld, Warren~D. Smith, and Patrick
  White.
\newblock Testing {Closeness} of {Discrete} {Distributions}.
\newblock {\em arXiv:1009.5397 [cs, math, stat]}, September 2010.
\newblock arXiv: 1009.5397.

\bibitem{blais_property_2012}
Eric Blais, Joshua Brody, and Kevin Matulef.
\newblock Property {Testing} {Lower} {Bounds} via {Communication} {Complexity}.
\newblock {\em computational complexity}, 21(2):311--358, May 2012.

\bibitem{bonis_constraining_2015}
Annalisa~De Bonis.
\newblock Constraining the number of positive responses in adaptive,
  non-adaptive, and two-stage group testing.
\newblock {\em Journal of Combinatorial Optimization}, pages 1--34, September
  2015.

\bibitem{bravyi_quantum_2011}
S.~Bravyi, A.~W. Harrow, and A.~Hassidim.
\newblock Quantum {Algorithms} for {Testing} {Properties} of {Distributions}.
\newblock {\em IEEE Transactions on Information Theory}, 57(6):3971--3981, June
  2011.

\bibitem{canonne_testing_2015}
C.~Canonne, D.~Ron, and R.~Servedio.
\newblock Testing probability distributions using conditional samples.
\newblock 44(3):540--616.

\bibitem{canonne_testing_2014}
Cl{\'e}ment Canonne and Ronitt Rubinfeld.
\newblock Testing {Probability} {Distributions} {Underlying} {Aggregated}
  {Data}.
\newblock In Javier Esparza, Pierre Fraigniaud, Thore Husfeldt, and Elias
  Koutsoupias, editors, {\em Automata, {Languages}, and {Programming}}, number
  8572 in Lecture {Notes} in {Computer} {Science}, pages 283--295. Springer
  Berlin Heidelberg, July 2014.
\newblock DOI: 10.1007/978-3-662-43948-7\_24.

\bibitem{chakraborty_power_2016}
S.~Chakraborty, E.~Fischer, Y.~Goldhirsh, and A.~Matsliah.
\newblock On the {Power} of {Conditional} {Samples} in {Distribution}
  {Testing}.
\newblock {\em SIAM Journal on Computing}, pages 1261--1296, January 2016.

\bibitem{chakraborty_new_2010}
Sourav Chakraborty, Eldar Fischer, Arie Matsliah, and Ronald de~Wolf.
\newblock New {Results} on {Quantum} {Property} {Testing}.
\newblock {\em arXiv:1005.0523 [quant-ph]}, May 2010.
\newblock arXiv: 1005.0523.

\bibitem{chan_optimal_2014}
Siu-On Chan, Ilias Diakonikolas, Gregory Valiant, and Paul Valiant.
\newblock Optimal {Algorithms} for {Testing} {Closeness} of {Discrete}
  {Distributions}.
\newblock In {\em Proceedings of the {Twenty}-{Fifth} {Annual} {ACM}-{SIAM}
  {Symposium} on {Discrete} {Algorithms}}, {SODA} '14, pages 1193--1203,
  Philadelphia, PA, USA, 2014. Society for Industrial and Applied Mathematics.

\bibitem{cleve1998quantum}
Richard Cleve, Artur Ekert, Chiara Macchiavello, and Michele Mosca.
\newblock Quantum algorithms revisited.
\newblock In {\em Proceedings of the Royal Society of London A: Mathematical,
  Physical and Engineering Sciences}, volume 454, pages 339--354. The Royal
  Society, 1998.

\bibitem{de_bonis_optimal_2005}
A.~De~Bonis, L.~Gasieniec, and U.~Vaccaro.
\newblock Optimal {Two}-{Stage} {Algorithms} for {Group} {Testing} {Problems}.
\newblock {\em SIAM Journal on Computing}, 34(5):1253--1270, January 2005.

\bibitem{deutsch_rapid_1992}
David Deutsch and Richard Jozsa.
\newblock Rapid {Solution} of {Problems} by {Quantum} {Computation}.
\newblock {\em Proceedings of the Royal Society of London A: Mathematical,
  Physical and Engineering Sciences}, 439(1907):553--558, December 1992.

\bibitem{diakonikolas_testing_2015}
Ilias Diakonikolas, Daniel~M. Kane, and Vladimir Nikishkin.
\newblock Testing {Identity} of {Structured} {Distributions}.
\newblock In {\em Proceedings of the {Twenty}-{Sixth} {Annual} {ACM}-{SIAM}
  {Symposium} on {Discrete} {Algorithms}}, {SODA} '15, pages 1841--1854,
  Philadelphia, PA, USA, 2015. Society for Industrial and Applied Mathematics.

\bibitem{goldreich_property_2010}
Oded Goldreich.
\newblock {\em Property {Testing}: {Current} {Research} and {Surveys}}.
\newblock Springer, October 2010.
\newblock Google-Books-ID: HIdqCQAAQBAJ.

\bibitem{goldreich_property_1998}
Oded Goldreich, Shari Goldwasser, and Dana Ron.
\newblock Property {Testing} and {Its} {Connection} to {Learning} and
  {Approximation}.
\newblock {\em J. ACM}, 45(4):653--750, July 1998.

\bibitem{grimmett2014probability}
Geoffrey Grimmett and Dominic Welsh.
\newblock {\em Probability: an introduction}.
\newblock Oxford University Press, 2014.

\bibitem{guha_streaming_2006}
Sudipto Guha, Andrew McGregor, and Suresh Venkatasubramanian.
\newblock Streaming and {Sublinear} {Approximation} of {Entropy} and
  {Information} {Distances}.
\newblock In {\em Proceedings of the {Seventeenth} {Annual} {ACM}-{SIAM}
  {Symposium} on {Discrete} {Algorithm}}, {SODA} '06, pages 733--742,
  Philadelphia, PA, USA, 2006. Society for Industrial and Applied Mathematics.

\bibitem{montanaro_survey_2013}
Ashley Montanaro and Ronald de~Wolf.
\newblock A {Survey} of {Quantum} {Property} {Testing}.
\newblock {\em arXiv:1310.2035 [quant-ph]}, October 2013.
\newblock arXiv: 1310.2035.

\bibitem{odonnell_quantum_2015}
Ryan O'Donnell and John Wright.
\newblock Quantum {Spectrum} {Testing}.
\newblock {\em arXiv:1501.05028 [quant-ph]}, January 2015.
\newblock arXiv: 1501.05028.

\bibitem{robbins1955remark}
Herbert Robbins.
\newblock A remark on stirling's formula.
\newblock {\em The American Mathematical Monthly}, 62(1):26--29, 1955.

\bibitem{schatzman_numerical_2002}
M.~Schatzman and Michelle Schatzman.
\newblock {\em Numerical {Analysis}: {A} {Mathematical} {Introduction}}.
\newblock Clarendon Press, 2002.
\newblock Google-Books-ID: 3SuNiR1hzxUC.

\bibitem{sykora_quantum}
Stanislav S{\'y}kora.
\newblock Quantum theory and the bayesian inference problems.
\newblock {\em Journal of Statistical Physics}, 11(1):17--27.

\bibitem{valiant_power_2011}
G.~Valiant and P.~Valiant.
\newblock The {Power} of {Linear} {Estimators}.
\newblock In {\em 2011 {IEEE} 52nd {Annual} {Symposium} on {Foundations} of
  {Computer} {Science} ({FOCS})}, pages 403--412, October 2011.

\end{thebibliography}
